\pdfoutput=1
\documentclass[12pt]{article}
\usepackage{amssymb}
\usepackage{amsthm}
\usepackage{multirow}
\usepackage{amsthm,amsmath,amsfonts,natbib}
\usepackage{verbatim}
\usepackage{graphicx}
\usepackage{caption}
\usepackage{tabu}
\usepackage{booktabs}
\usepackage[font=footnotesize]{caption}
\usepackage{subcaption}
\usepackage{bm}
\usepackage{bbm}
\usepackage{xcolor}
\usepackage{bbm}
\usepackage{float}
\floatstyle{plaintop}
\restylefloat{table}
\RequirePackage[colorlinks,citecolor=blue,urlcolor=blue]{hyperref}
\usepackage{natbib}
\usepackage{url} 

\newcommand{\blind}{1}

\newtheorem{theorem}{Theorem}
\newtheorem{definition}{Definition}
\newtheorem{property}{Property}
\newcommand\independent{\protect\mathpalette{\protect\independenT}{\perp}} \def\independenT#1#2{\mathrel{\rlap{$#1#2$}\mkern2mu{#1#2}}}
\newcommand{\dr}{\textnormal{d}}

\begin{document}

\def\spacingset#1{\renewcommand{\baselinestretch}%
{#1}\small\normalsize} \spacingset{1}


\if1\blind
{
 \title{\bf A conditional independence framework for coherent modularized inference}
  \author{Manuele Leonelli \\
   School of Mathematics and Statistics, University of Glasgow
\\
\\
    Martine J. Barons and Jim Q. Smith \\
   Department of Statistics, University of Warwick}
  \maketitle
} \fi

\if0\blind
{
  \bigskip
  \bigskip
  \bigskip
  \begin{center}
    {\LARGE\bf A changepoint approach for the identification of financial extreme regimes}
\end{center}
  \medskip
} \fi

\bigskip
\begin{abstract}
Inference in current domains of application are often complex and require us to integrate the expertise of a variety of disparate panels of experts and models coherently. In this paper we develop a formal statistical methodology to guide the networking together of a diverse collection of probabilistic models. In particular, we derive sufficient conditions that ensure inference remains coherent across the composite before and after accommodating relevant evidence.
\end{abstract}

\noindent%

\spacingset{1.45} 
\section{Introduction}
\label{sec:intro}
For large scale applications is now often necessary to integrate many probabilistic components - themselves often massive - into a single coherent system. One of the earliest examples comes from nuclear emergency management. There engineering, athmospheric, economic and social components needed to be integrated to give uncertainty measures of the possible unfoldings of a crisis \citep{French1991}. Then expedience dictated that point estimates of outputs from one component was donated to an adjacent component whilst the uncertainty associated with such estimates was simply discarded. This naive approach of course dramatically affects the coherence of the composite system both for decision making \citep{Leonelli2013,Leonelli2015} and inference \citep{Plummer2015}. Despite of the dangers associated to modular approaches, such practice is still often used, for instance in climate modelling \citep{Beal2010}, air pollution  \citep{Blangiardo2011}, econometrics \citep{Murphy2002} and epidemiology \citep{Li2017}.

Exceptions of course exist. For example some early work for the full integration of different components was developed in artificial intelligence \citep{Johnson2012,Mahoney1996,Xiang2002} although this work exclusively focused on  Bayesian networks. Only recently has the statistical community recognized the need for a more general study of how and when such an integration can be coherently performed \citep{Liu2009}. In the context of graphical models, Markov combinations \citep{Massa2010,Massa2017} and Markov melding \citep{Goudie2018} are coherent frameworks designed to compose models which share some common features. \citet{Jacob2017} characterizes the cases where a modular approach may outperform the inference resulting from a full joint statistical specification.

Here we consider the situation where inference can be exclusively carried out in a modular fashion, where different panels of experts are responsible for their own component of the system which they update autonomously. In the above-mentioned applications, the fields of expertise that need to be brought together are so heterogeneous that no-one could realistically \lq\lq{o}wn\rq\rq{} all the statements about the fully joint distribution unless there existed shared structural assumptions. Even if it were feasible to build a single joint model, typically the components need to be continually revised to accommodate new undersanding, science and data. Any overaching probabilistic model would thus quickly become obsolete since the judgements it embodied would no longer reflect current understanding.

Here we determine when and how expert judgement and data in such a modular framework can be formally and coherently combined, taking into account of the relative strength of different evidence. We demonstrate that under a set of structural assumption often met in practice, inference can be guaranteed to be coherent and also show that the beliefs it embodies can be structured in a way that incorporates  all usable, informed expert judgment and data drawn together from the different components. Furthermore, in a sense we define later, if an appropriate protocol guiding the nature and quality of the data input is in place, then inference can be distributed, meaning that that the system remains coherent when each panel  autonomously  updates its beliefs about its domain of expertise. Distributivity gives the added benefit that any output probability associated with the composite can be typically calculated very quickly \citep{Leonelli2015}.

\section{The conditional independence framework}
\label{sec:framework}

We start setting up a formal framework for the combination of panels' judgements and components. A large vector of random variables, $\bm{Y}$, measures various features of an unfolding future. Suppose that different components of this vector are evaluated and overseen by $m$ different panels of domain experts, $ G_{1},\dots, G_{m} $. Panel $G_{i}$ is responsible for the output vector $\bm{Y}_i$, where $\bm{Y}_i$ takes values in $\mathbb{Y}_i$, $i\in[m]=\{1,\dots,m\}$, and $\bm{Y}=( \bm{Y}_{i})_{i\in[m]} $.  The implicit (albeit virtual) owner of the beliefs of the panels is henceforth referred to as the supraBayesian, meaning that the group reasons as a single person would and it is her coherence that we are concerned with. 

A modular system should endeavour to accommodate probabilistic information provided by the most well-informed experts only. Therefore each panel only donates beliefs associated with their own domain of expertise and not their beliefs about the whole vector $\bm{Y}$. Each panel $G_{i}$ is asked to give the supraBayesian their beliefs about the probability distribution of $ \bm{Y}_{i}$ over which $G_i$ has oversight, conditional on certain measurable functions $\left\{ A_{i}( \bm{Y}):A_{i}\in \mathbb{A}_i\right\} $, where $\mathbb{A}_{i}$ could be null. For instance $\mathbb{A}_{i}$ could be the set of different possible combinations of the covariates' levels on which $\bm{Y}_i$ might depend. 

Henceforth we assume that all panels make their inferences in a parametric setting where $\bm{Y}$ is parametrised by $\bm{\theta }=\left( \bm{\theta}_i\right)_{i\in[m]} \in \Theta$, and $\bm{\theta }_{i}$ parametrises  $G_i$'s sample distributions. When the parameter space can be written as a product space, $\Theta = \times_{i\in[m]}\Theta _{i}$,  where $\Theta_i$ is $G_i$'s parameter space, panels are variationally independent \citep{Dawid1993}. We  assume this property holds, as often in practice. This is in fact a necessary condition for inference to be distributed.

In this parametric setting panels  have two quantities available to them. The first are sample distibutions over future measurements for which they have responsibility 
$
\Pi _{i}^{y|\theta }\triangleq \left\{f_i(\bm{Y}_{i}\ |\ A_i, \bm{\theta }_{i}):\bm{\theta }_{i} \in \Theta _{i}, A_i \in \mathbb{A}_i\right\}
$. For example, if $\bm{Y}$ were discrete and finite, then each panel might be asked to provide certain multi-way conditional probability tables over their subvector $\bm{Y}_{i}$, conditional on each $A_{i}\in \mathbb{A}_{i}$. In this case $\bm{\theta }_{i}$ would be the concatenated probabilities within all these tables. The second are panel densities
$
\Pi _{i}^{\theta }\triangleq \left\{ \pi _{i}\left( \bm{\theta }_{i}\ |\ A_i\right):A_i \in \mathbb{A}_i\right\}$, such as a joint probability distribution over all the probabilities specified in the multi-way conditional probability tables above. Of course panels may want to use various data available to them to infer the distribution of $\bm{\theta}_i$, $i\in[n]$. However each panel typically performs this inference autonomously. 
\begin{definition}
If the supraBayesian's beliefs are functions of the autonomously calculated panels' beliefs we say that inference is distributed.
\end{definition}
If $G_i$'s inferences are going to be inherited by the supraBayesian, it is by no means automatic that such autonomous updating is justified.  We henceforth assume Property \ref{prop} holds.

\begin{property}
\label{prop}
All agree on the variables $\bm{Y}$ defining the process together with qualitative statements about the dependence between $\bm{Y}$ and $\bm{\theta}$. Call these assumptions the structural consensus.
\end{property}

This consensus can often be expressed through an appropriately chosen graphical or conditional independence structure both across  the distribution of $ \bm{Y}\ |\ \bm{\theta } $, and also that of $\bm{\theta }$.  These types of assumptions are often complex but can be represented by a variety of well-known frameworks \citep{Leonelli2015a}. Technically we can think of the structural consensus as the qualitative beliefs that are shared as common knowledge by all the panel members. These are usually expressed in common  language and so are more likely to form part of the structural consensus \citep{Smith1996}.  For the coherence of modular systems, inference also needs to be sound.

\begin{definition}
If by adopting the structural consensus the supraBayesian underlying beliefs about a domain overseen by a panel $G_{i}$ are $\{ \Pi _{i}^{y|\theta },\Pi _{i}^{\theta}\}$, $i\in[m]$, we say that inference is sound.
\end{definition}

\section{Illustrations of some of the inferential challenges}
\label{sec:illustrations}
We next give a flavour of the difficulties associated with the integration of the panels beliefs and the dangers of modular inference when not justified. Consider the simplest scenario where  the number of panels $m=2$ and the structural consensus specifies that $\bm{Y}= (Y_{1},Y_{2})$ are binary. Here the random variable $Y_{1}$ is an indicator of whether or not a foodstuff has become poisonous and  $Y_{2}$ is an indicator of whether or not contamination is detected before the food is distributed to the public. Sample distributions $\Pi_{i}^{y|\theta }$ given by expert panel $G_{i}$ are saturated so that $\theta_{i}\triangleq P(Y_{i}=1)$. Suppose the structural consensus further specifies that $Y_{2}\independent Y_{1}\;|\;(\theta_{1},\theta _{2})$.  This implies that, given the probabilities $\theta_1$ and $\theta_2$, observing poisonous foodstuff does not affect our beliefs about the detection regime. Suppose the main inferential object of interest is 
 $\mathbb{E}(Y_{1}Y_{2})$, an indicator of whether the public is exposed to the contamination. To derive such score the supraBayesian needs to able to calculate $\mathbb{E}\left( \theta_{1}\theta _{2}\right)$. This can be derived from the panel's densities if and only if the independence $\theta _{1}\independent \theta _{2}$ was in the structural consensus. With this additional condition, assume a random vector $\left( \bm{X}_{1},\bm{X}_{2}\right) $ is sampled from the same population as $\left( Y_{1},Y_{2}\right) $ so that each panel updates  its parameter densities from $\pi _{i}(\theta _{i}\ )$ to $\pi_{i}(\theta _{i}\ |\ \bm{x}_{i})$, $i\in[2]$, using its own separate randomly sampled populations, $\bm{x}_{i}$, concerning $\theta_i$ alone. The supraBayesian inference is then sound and distributed. By adopting all these beliefs as her own, she acts as if she had sight of all the available information and had processed this information herself.

Note that in the example above  the  assumption $\theta_1\independent\theta_2$ is critical for distributivity. Suppose on the contrary that  $\pi (\theta _{2}\ |\ \theta _{1})$  needs to be a function of $\theta _{1}$. Then
$
\pi _{2}(\theta _{2}\ |\ \bm{x}_{1},\bm{x}_{2})=\int_{0}^{1}\pi _{2}(\theta _{2}\ |\ \theta_{1},\bm{x}_2)\pi_1(\theta _{1}\ |\ \bm{x}_1)\dr\theta _{2},
$
 where the prior dependence of $\theta_2$ on $\theta_1$ induces a dependence of $\theta_2$ on $\bm{x}_1$. So  $\pi _{2}(\theta _{2}\ |\ \bm{x}_{1},\bm{x}_{2})\neq \pi _{2}(\theta _{2}\ |\ \bm{x}_{2})$ in general. By devolving inference to the two panels who learn autonomously,  the supraBayesian cannot act as a single Bayesian would by using $\pi _{2}(\theta _{2}\ |\ \bm{x}_{2})$. It follows that the system is no longer sound, although when supporting evidence remains unseen the supraBayesian  appears to act coherently.  She is, implicitly,  assuming that $\theta_1 \independent \theta_2$, which is contrary to the reasoning $G_2$ would want to provide.

Perhaps more importantly, even if global independence is justified a priori, the assumption that data collected by the two panels and individually used to adjust their beliefs does not inform both parameters is also a critical one.  Suppose that $G_{1}$ and $G_{2}$ both see their respective margin concerning the experiment in Table \ref{experiment} and each uses this experiment to update its  marginal distribution on $\theta _{i}$, $i\in[2]$.
\begin{table}
\[
\begin{array}{ccccc}
Y_{1}\backslash Y_{2} & 0 & 1 &  &  \\ 
0 & 5 & 45 & 50 & n-x_{1} \\ 
1 & 45 & 5 & 50 & x_{1} \\ 
& 50 & 50 & 100 &  \\ 
& n-x_{2} & x_{2} &  &
\end{array}
\]
\caption{Experiment of the poisonous food example in Section \ref{sec:illustrations}.\label{experiment}}
\end{table} 
If both began with a prior symmetric about $0.5$ and the structural consensus included $\theta _{1}\independent \theta _{2}$ then the supraBayesian would assign $\mathbb{E}(Y_1Y_2\ |\ x_1,x_2)=\mathbb{E}(Y_1\ |\ x_1)\mathbb{E}(Y_2\ |\ x_2)=0.25$.
This inference contrasts with inferences the supraBayesian would make on seeing the whole table and assuming $\theta _{1}\independent \theta _{2}$ a priori. With a fairly uninformative prior on the two margins, her judgement would be approximately $0.05$, five times smaller than above. If conversely it were only possible to see the table of randomly sampled counts associated with $Y_{1}Y_{2}$ and the supraBayesian used this information directly, for example by introducing a uniform prior on $ P(Y_1Y_2=1)$, then her  posterior mean would be approximately $0.05$. However, observations have made the  independence $\theta_1\independent\theta_2$ no longer formally valid and the distributivity of the system is destroyed if this data is accommodated. 

Thus considerable care needs to be exercised before a modular system can be expected to work reliably, even in the simplest of networks.  In the next section we  prove some conditions which ensure the supraBayesian inference is distributed and sound.  

\section{Coherence and modularization}
\subsection{Information and admissibility protocols}
\label{ssec:admis}
Suppose panels are progressively informed by new collections of evidence.  In practice, within the totality of information conceivably available at time $t$, usually only a subset, the admissible evidence, is of sufficient quality and has suitable form to be used. The sort of information excluded or delayed admittance might include evidence whose relevance is ambiguous or of a type which might introduce insurmountable computational challenges. An admissibility protocol is therefore needed to avoid the difficulties illustrated in Section \ref{sec:illustrations}. 

Let $I_{0}^{t}$ denote all the admissible evidence which is common knowledge to all panel members at time $t$. Let $I_{ij}^{t}$ denote the subset of the admissible evidence panel $G_i$ would use at time $t$ if acting autonomously to assess their beliefs about $\boldsymbol{\theta }_{j}$,  $i,j\in[m]$.  We define the admissible evidence as $I_{+}^{t} \triangleq \{ I_{ij}^{t}:i,j\in[m]\}\cup I_0^t$ and $I_{\ast}^{t}\triangleq \{ I_{ii}^{t}:i\in[m]\}$ to be the subset of the admissible evidence each panel $G_i$ would use to update $\bm{\theta}_i$, $i\in[m]$.

Although not explicitly stated, protocols for the selection of good quality evidence are often used  in several domains. A notable example is the Cochrane systematic review \citep{Higgins2008} developed to pare away information which might be
ambiguous and potentially distort evidence concerning medical treatments, through a trusted set of principles relevant to the domain. This may seem restrictive, but in practical applications the need to be selective about experiments that can provide evidence of an acceptable quality is  widely acknowledged. 

For a modular system, the needs for such an admissibility protocol are even more important
because of its collective structure. So we assume that panels, both individually and corporately, agree on an appropriate protocol for selecting suitable evidence, mirroring Cochrane reviews in ways relevant to their domain. However, for our particular purposes one additional requirement is needed in this setting: the chosen admissibility protocol must also ensure distributivity. For we have already argued that, if this is not the case, then the supraBayesian inference is either dependent on arbitrary assumptions and difficult to calculate or, if distributivity is forced, incoherent. 

\subsection{Conditional independence conditions for coherence}

Assume that all panellists, as represented collectively by the supraBayesian, agree that their inferences should obey the (qualitative) semi-graphoid axioms. These axioms are widely accepted as appropriate for reasoning about evidence when irrelevance statements are read as conditional independences: Bayesian systems always respect these properties \citep{Studeny2006}. The four properties below ensure the soundness of a supraBayesian inference within an admissibility protocol. Although these are in no sense automatic, they are, nevertheless,  satisfied by a very diverse collection of models and information sets and can be checked for their plausibility in any given context \citep{Leonelli2015a}. Let $\bm{\theta }_{i^{-}}\triangleq \left( \theta _{j}\right)_{j\in[m]\setminus \{i\}}$.  

\begin{definition}
\label{def:cond}
An admissibility protocol is said delegable at time $t$ if 
\begin{equation}
I_{+}^{t}\independent \boldsymbol{\theta }\ |\ I_{0}^{t},I_{\ast }^{t},
\label{delegatable}
\end{equation}
separately informed at time $t$ if 
\begin{equation}
I_{ii}^{t}\independent \boldsymbol{\theta }_{i^{-}}\ |\ I_{0}^{t},\boldsymbol{\theta }_{i}, \label{sep inform}
\end{equation}
cutting at time $t$ if 
\begin{equation}
I_{\ast }^{t}\independent \boldsymbol{\theta }_{i}\ |\ I_{0}^{t},I_{ii}^{t},
\boldsymbol{\theta }_{i^{-}}, \label{cutting}
\end{equation}
commonly separated at time $t$ if 
\begin{equation}
\independent _{i\in[m]}\boldsymbol{\theta }_{i}\ |\ I_{0}^{t}.  \label{commonly sep}
\end{equation}
\end{definition}

 An admissibility protocol is delegable when the admissible evidence $I_{+}^{t}$ is the union of the evidence $I_{0}^{t}$ shared by all panels plus the evidence $I_{\ast}^{t}$ each panel has about its own domain of expertise. This condition is ensured if panels are working collaboratively rather than competitively. Alternatively the protocol could itself simply demand that $I_{+}^{t}=\left\{I_{0}^{t},I_{\ast }^{t}\right\} $.
 
  For the separately informed condition, pieces of evidence $G_{i}$ might collect individually cannot be informative about $\bm{\theta}_{i^{-}}$ once the domain experts' evidence has been fed in. This is satisfied when a new piece of information could only be added to $ I_{ii}^{t}$  if the evidence it provided about $\boldsymbol{\theta }_{i}$ would not depend on $\boldsymbol{\theta }_{i^{-}}$. Although commonly valid \citep{Leonelli2015a}, one context where  it is violated is when sampling is not ancestral \citep{Smith2010}: sometimes due to a hidden confounder across the different models of the system.

For cutting, once $\left\{ I^t_{0},I_{ii}^{t}\right\} $ is known, no-one believes that a panel $G_{j}$ has used any information that $ G_{i}$ might also want to use to adjust its beliefs about $\boldsymbol{ \theta }_{i}$. This captures what we might mean by `panel of experts'. One situation when the cutting condition is violated is if  a panel $G_{i}$ marginalises out a parameter in $\boldsymbol{\theta }_{j}$ to accommodate a piece of evidence, making $\bm{\theta}_i$ and $\bm{\theta}_j$ dependent on each other. 

When parameters are commonly separated the information that everyone shares separates the parameters. We note that this  is nearly always  assumed in the context of practical applications of graphical models \citep{Dawid1993}. 

The following result, analogous to \citet{Goldstein1996} about the use of linear Bayes methods and a single agent, can now be proved. 

\begin{theorem}
Under structural consensus, the supraBayesian inference is sound and distributed if the admissibility protocol is delegable, separately informed, cutting and commonly separated. 
\label{theo:gold}
\end{theorem}

\begin{proof}
The result follows if, at any time $t$, the supraBayesian holds panel independent beliefs
\begin{equation}
\bm{\theta }_{i}\independent \bm{\theta }_{i^{-}}\;|\;I_{+}^t,
\label{indep components}
\end{equation}
and if, when assessing $\bm{\theta }_{i}$, she only uses the information $G_{i}$ would use if acting autonomously
\begin{equation}
\bm{\theta }_{i}\independent I_{+}^t\;|\;I_{0}^t,I_{ii}^t.  \label{indep updating}
\end{equation}
These conditions are then sufficient for soundness and distributivity. Because even if all panels could share each other's information then they would come to the same assessment. 

To prove condition (\ref{indep components}) note that from common separability it follows that $\bm{\theta}_i \independent \bm{\theta}_{i^{-}} \;|\; I_{0}^t$, which combined with the separately informed condition in equation (\ref{sep inform}) through perfect decomposition and using the symmetric property of semi-graphoids axioms gives $
\bm{\theta}_{i^{-}} \independent I_{ii}^t,\bm{\theta}_i\;|\; I_{0}^t$.
Using perfect decomposition and symmetry again, it follows that $\bm{\theta}_i\independent \bm{\theta}_{i^{-}}\;|\; I_{ii}^t,I_{0}^t$. When combined with the cutting condition in (\ref{cutting}) via perfect decomposition this gives 
\begin{equation}
\label{proof:sd}
\bm{\theta}_i\independent I_*^t,\bm{\theta}_{i^{-}}\;|\; I_{0}^t,I_{ii}^t.
\end{equation}
 Then again by perfect decomposition and since $I_{ii}^t$ is a function of $I_*^t$ we have that
\begin{equation}
\label{proof:sound2}
\bm{\theta}_i\independent \bm{\theta}_{i^{-}}\;|\; I_{0}^t, I_{*}^t.
\end{equation}
Using perfect decomposition, the delegatable condition in equation (\ref{delegatable}) can be written as
\begin{equation}
I_+^t\independent \bm{\theta}_{i}\;|\; I_{0}^t,I_*^t,\bm{\theta}_{i^{-}}. \label{proof:sound3}
\end{equation}
Combining via perfect decomposition equations (\ref{proof:sound2}-\ref{proof:sound3}) and using symmetry we have that $\bm{\theta}_{i}\independent I_+^t,\bm{\theta}_{i^{-}}\;|\; I_{0}^t,I_*^t$, so by perfect decomposition it follows that $
\bm{\theta}_{i}\independent \bm{\theta}_{i^{-}}\;|\; I_{0}^t,I_*^t,I_+^t$, 
which, since $I_*^t$ and $I_{0}^t$ are functions of $I_+^t$, can be written as equation (\ref{indep components}). 

To show that equation (\ref{indep updating}) holds, note that another implication of delegatability in equation (\ref{delegatable}) by perfect decomposition is that, again using the fact that $I_{ii}^t$ is a function of $I_*^t$,
\begin{equation}
\bm{\theta}_i\independent I_{+}^t\;|\; I_{0}^t,I_{ii}^t,I_*^t.\label{proof:sound41}
\end{equation}
Noting that by perfect decomposition equation (\ref{proof:sd}) implies that
\begin{equation}
\bm{\theta}_i\independent I_*^t \;|\; I_{0}^t,I_{ii}^t, \label{proof:sound5}
\end{equation}
from equations (\ref{proof:sound41}) and (\ref{proof:sound5}) by perfect decomposition it follows that $\bm{\theta}_i  \independent I_+^t,I_*^t\;|\; I_{0}^t,I_{ii}^t,$. This, since $I_*^t$ is a function of $I_+^t$, is equivalent to equation (\ref{indep updating}).
\end{proof}

\subsection{Likelihood separation for modularized inference}
\label{sec:lik}
We now focus our attention onto  probabilistic systems. Suppose that the supraBayesian believes  $\independent _{i\in[m]}\boldsymbol{\theta }_{i}\ |\ I_{+}^{0}$ and that the only evidence available to panels is in the form of datasets $ \boldsymbol{x}^{t}=\left\{ \boldsymbol{x}_{\tau }:\tau \leq t\right\} $ which then populate  $I_{+}^{t}$. The features that ensure soundness and distributivity can be expressed in terms of the separability of a likelihood. Let $l(\boldsymbol{\theta }\ |\ \boldsymbol{x}^{t})$ denote a likelihood over the parameter $\boldsymbol{\theta }$ of the distribution of $\boldsymbol{Y}$ given $\boldsymbol{x}^{t}$. 

\begin{definition}
An admissibility protocol is said panel separable if
\begin{equation*}
l(\boldsymbol{\theta }\ |\ \boldsymbol{x}^{t})=\prod_{i\in[m]}l_{i}(\boldsymbol{\theta }_{i}\ |\ \boldsymbol{t}_{i}(\boldsymbol{x}^{t})),
\end{equation*}
where $l_{i}(\boldsymbol{\theta }_{i}\ |\ \boldsymbol{t}_{i}(\boldsymbol{x}^{t}))$ is a function of $\boldsymbol{\theta }$ only through $\boldsymbol{\theta }_{i}$ and $\boldsymbol{t}_{i}(\boldsymbol{x}^{t})$ is a statistic of $\boldsymbol{x}^{t}$  formally accommodated by $G_{i}$ into $I_{ii}^{t}$ to form its own posterior assessment of $\boldsymbol{\theta }_{i}$, $i\in[m]$.
\end{definition}

We now have the following theorem that gives good practical guidance about when and how soundness and delegability can be preserved over time, as new datasets inform the composite.

\begin{theorem}
Under structural consensus, suppose the supraBayesian inference is sound and distributed at time $t=0$. Provided the prior over $\boldsymbol{\theta }$ is absolutely continuous with respect to Lebesgue measure, the supraBayesian inference is sound and distributed at time $t>0$ if and only if the admissibility protocol is panel separable. \label{theo:seplik}
\end{theorem}

\begin{proof}
Under the initial hypotheses, by Theorem \ref{theo:gold}, the prior joint density can be written in a product form  $\pi (\bm{\theta })=\prod_{i\in[m]}\pi _{i}(\bm{\theta }_{i}).$ It follows that under this admissibility protocol
$
\pi (\bm{\theta }\;|\;\bm{x}^{t})=\prod_{i\in[m]}\pi_{i}(\bm{\theta }_{i},\bm{t}_{i}(\bm{x}^{t}))$, 
where $\pi _{i}(\bm{\theta }_{i},\bm{t}_{i}(\bm{x}^{t}))\propto l_{i}(\bm{\theta }_{i}\;|\;\bm{t}_{i}(\bm{x}^{t}))\pi _{i}(\bm{\theta }_{i}).$ By hypothesis $\pi _{i}(\bm{\theta }_{i},\bm{t}_{i}(\bm{x}^{t}))$ is adopted by the supraBayesian, thus guaranteeing soundness. In particular we have that 
\begin{equation*}
\independent_{i\in[m]}\bm{\theta }_{i}\;|\;I_{0}^{t},I_{+}^{t}\;\;\Longleftrightarrow\; \;\independent_{i\in[m]}\bm{\theta }_{i}\;|\;I_{0}^{t},I_{\ast }^{t},
\end{equation*}
where $I_{+}^{t} =\left\{I_{0}^{t},I_{\ast}^{t}\right\} =\left\{ I_{0}^{0},I_{\ast }^{0},\bm{x}^{t}\right\}$ and, for $i\in[m]$, $\left\{ I_{0}^{t},I_{ii}^{t}\right\} =\left\{ I_{0}^{t},I_{ii}^{0},\bm{t}_{i}(\bm{x}^{t})\right\}$. Since by definition $\bm{t}_{i}(\bm{x}^{t})$ is known to $G_{i}$, $i\in[m]$,  the system is also distributed. 

Finally note that if $
l(\bm{\theta }\;|\;\bm{x}^{t})\neq\prod_{i\in[m]}l_{i}(\bm{\theta }_{i}\;|\;\bm{t}_{i}(\bm{x}^{t}))$ on a set $A$ of non zero prior measure, then the posterior density $\pi_{A}(\bm{\theta }\;|\;\bm{x}^{t})$ on $A$ has the property that, for all $\bm{\theta} \in A$,
$
\pi _{A}(\bm{\theta }\;|\;\bm{x}^{t})\neq\prod_{i\in[m]}\pi _{A,i}(\bm{\theta }_{i},\bm{t}_{i}(\bm{x}^{t})),
$ 
where $\pi_{A,i}$ denotes the density delivered by panel $G_i$ for the parameters it oversees in $A$. This means that panel parameters are a posteriori dependent. So in particular the density determined by the margins is not sound.
\end{proof}

By designing  a single experiment to be orthogonal over parameters $\boldsymbol{\theta }_{i}$  and $\boldsymbol{\theta }_{j}$, overseen by $G_i$ and $G_j$ respectively, inference is sound and distributed under the conditions of the theorem. From this single experiment  data can still be included in the admissible evidence. Note, however, that the converse demonstrates that some protocol is certainly needed to preserve  distributivity. 

\section{Conclusion}
The complexity of modern applications often require inference to be  modular. Here we have identified a set of conditions under which autonomous component distributions can be aggregated coherently to provide a unique probabilistic assessment. Many modelling tools commonly used, often in the form of a graphical model, can be set up to respect such conditions \citep{Leonelli2015a}. The novel framework we introduced here is currently being used in practice to address the critical application of food security, where a suite of diverse components need to be networked together in a coherent fashion to support policy makers \citep{Barons2018}. Interestingly when such systems are only used to inform a formal Bayesian decision analysis, it is possible to show that the conditions that lead to soundness and distributivity can be made  milder than the ones required in the fully inferential settings discussed in Theorems \ref{theo:gold} and \ref{theo:seplik} \citep{Leonelli2017}.

Of course, in some situations the likelihood separation needed to ensure the enduring formal distributivity breaks down. Then we need to fall back on approximate inferential methods that preserve distributivity. However preliminary results suggest that sensible approximate methods, whose form is similar to variational Bayes methods, can still work effectively provided that statistical diagnostics are employed to check that the structural consensus remains plausible in the light of the information available. 

\section*{Acknowledgement}
Research supported by EPSRC, JQS \& MJB were funded by EPSRC grant EP/K039628/1; JQS also supported by the Alan Turing Institute under EPSRC grant EP/N510129/.

\bibliographystyle{Chicago}
\bibliography{Bib.bib}

\end{document}